\documentclass[12pt]{article}
\usepackage{ae} 
\usepackage[T1]{fontenc}
\usepackage{amssymb}
\usepackage{latexsym}
\usepackage{amsmath}
\usepackage[ansinew]{inputenc}
\usepackage{euscript}
\usepackage{epsfig}
\usepackage{multirow}
\usepackage{anysize}

\newcommand{\qed}{\hspace*{\fill}
            $\Box$\smallskip}
\newcommand{\thmqed}{\hspace*{\fill}
            $\blacksquare$\smallskip}

\newenvironment{proof}{\noindent {\bf Proof:} \par}
                      {\qed}

 \newtheorem{theorem}{Theorem}
 \newtheorem{corollary}{Corollary}[section]
 \newtheorem{lemma}{Lemma}[section]
 \newtheorem{example}{Example}[section]
 \newtheorem{remark}{Remark}[section]
 \newtheorem{definition}{Definition}[section]
 \newtheorem{proposition}{Proposition}[section]
\begin{document}
\title{Consumer Search with Chain Stores \footnote{The Author has benefited from insightful comments from Prof. Rady, Prof. Janssen, Prof. Felbermayr and Prof. Holzner}}
\author{Sergey Kuniavsky\footnote{Munich Graduate School of Economics, Kaulbachstr. 45, Munich. Email: Sergey.Kuniavsky@lrz.uni-muenchen.de. Financial support from the Deutsche Forschungsgemeinschaft through GRK 801 is gratefully acknowledged.}}
\maketitle

\begin{abstract}
The paper explores a consumer search setting where the sellers have asymmetries. The model is an extension of the popular Stahl Model, which is widely used in the literature. The extension introduces sellers with heterogeneous stores number, reflecting the typical market structure. The market consists of several sellers heterogeneous in size consumers, some of which face a cost when sequentially searching. The paper shows that no symmetric model exist in the extension and asymmetric NE of the Stahl model are found for comparison. Additional results suggest that smallest sellers will be the ones offering lowest prices, in line with several real world examples provided in the paper. However, profits remain in most cases fixed per store, making a larger firm more profitable, yet with lower sold quantity. The findings suggest that on some level price dispersion will still exist, together with some level of price stickiness, both observed in reality.
\end{abstract}



\section{Introduction}

Empirical studies, such as \cite{HungMilk} or \cite{SpainBank}, have established that significant price dispersion exists even for homogeneous goods. As the literature suggests, this effect is observed in many market structures and is persistent. One of the explanations for this phenomenon is that consumers search for the cheapest price. Since searching is costly, consumers may settle down for a slightly higher price. In the literature many papers deal with search models, for example \cite{BJModel}, \cite{CM}, \cite{Stahl89} and \cite{Varian}. These models were developed originally in order to provide a solution to the Diamond Paradox \cite{Diamond}, which predicted a complete market failure. The search models vary in the scope, the length, the stopping condition or the information revealed during the consumer search. Additional Empiric studies, for example \cite{StahlOkEmp}, reveal that the model introduced by Stahl in \cite{Stahl89} perform very well and predicts correctly the pricing model of 86 out of 87 tested products. Moreover, \cite{ShoppersExplained} empirically shows the existence of the two consumer types predicted by this model. Therefore, this paper will concentrate on the Stahl search model.

An additional Phenomenon that can be observed, for example in \cite{HungMilk}, is a possible correlation between the price offered by a seller and the number of stores she has. Namely, the more branches a seller has the higher will be the price offered. Despite the fact that the Stahl Model has a variety of extensions, the literature dealing with asymmetries among sellers is not large. This is a very important extension, as in the real world the number of stores a seller has can vary, for example - see figure \ref{DNum}. Among the few papers in this field is \cite{AsymmetricSearch}, where only a model with 2 sellers is concerned. This paper shows several effects, which appear only when there are at least three sellers available. An additional paper is \cite{Monopolist}, where only a single large firm exists and all other are single store sellers. Already there it is noted that the larger firm charges a higher price. This paper investigates whether this is true in a more general setting than in \cite{Monopolist}. Here the Stahl model is extended, and each seller has a predefined, seller specific store number. The seller sets the same price for all stores (e.g.: Bank offers for saving accounts). The consumers search sequentially and uniformly among stores, rather than sellers. This implies that there is more chance for a consumer to turn up at a store belonging to the larger seller. Additionally, if a searcher is unsatisfied with a price, she would refrain from visiting any additional store of the previously visited seller.

\begin{figure}\label{DiscountersNum}
 \centering
    \includegraphics[height=80mm]{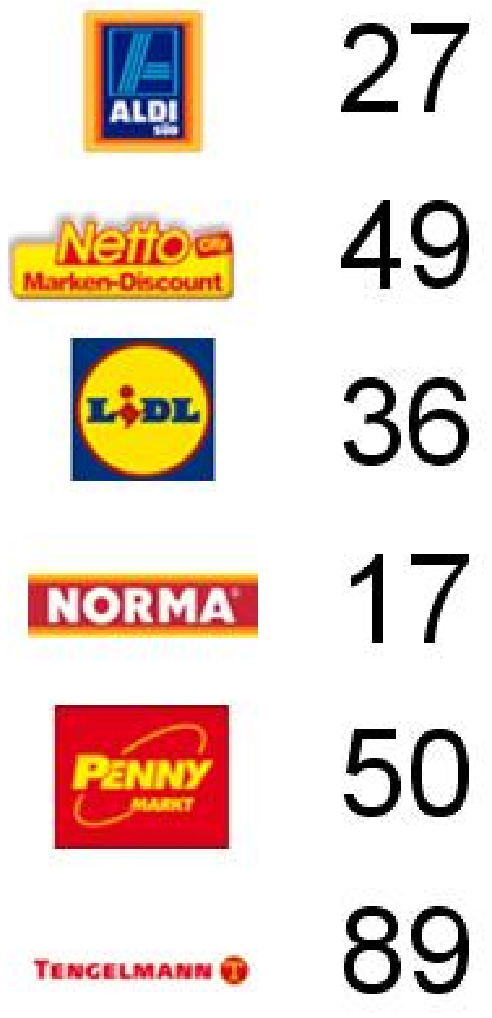}
  \caption{Number of discounter stores in Munich, according to kaufda.de, 5.2011}
  \label{DNum}
\end{figure}

The results overview is as follows:
\begin{enumerate}
\item Description of the asymmetric NE of the Stahl model
\item Formal introduction of an extension to the Stahl model, dealing with heterogeneous sellers
\item Description of the NE in the extended model, when the smallest seller is not unique
\item Description of the NE in the extended model, when the smallest seller unique
\item Examples for NE which illustrate each of the three cases
\end{enumerate}

The first thing one notice when discussing the extended Stahl model (with different size sellers) is a lack of a symmetric equilibrium, where all sellers use the same strategy. The original Stahl model has a unique symmetric Equilibrium, as shown in \cite{Stahl89} and the literature does not goes far beyond it. In order to provide an important building block for the extended Stahl model this paper first analyzes the asymmetric equilibria the original model has. For comparison, in another search model introduced in \cite{Varian}, it is shown in \cite{Tango} that there are asymmetric equilibria, but those can be ignored. In the Stahl model there might be additional equilibria when different settings are considered. For example, it is shown in \cite{BertrandEQ} that one can receive additional equilibria in commonly known games, when the scope is broadened. This paper finds a family of asymmetric equilibria to the original model, where strategies are of (at most) three types - some sellers (at least two) mix over the entire available price interval with a seller invariant distribution, whereas the second group (might be empty) selects the reserve price as a pure strategy. The third group (might be empty) has a pricing distribution which consists of a mass point at the reserve price, and use the same distribution as the first group up to a seller specific cutoff price. As for the extended model, the paper shows that when there are at least two smallest sellers, all sellers except the smallest sellers select the reserve price purely. The remaining sellers have a similar equilibrium to the original Stahl model, but with a lower portion of uninformed consumers (all those who visit one of the smallest sellers). This extends the result in \cite{Monopolist}, and shows that the lowest price will be obtained in one of the smaller chains. Moreover, in all equilibria found here (in both the original and extended models) all sellers have the same expected profit per store, all consumers buy at the first store they visit and no seller will ever set a price above the reserve price. Additional characteristic of the NE is that in the original model the expected profit for is equal for all sellers, and in the extended model, in most cases, the profit is a constant times the chain size (store number).

In case of a single smallest firm the equilibrium structure slightly changes. Now, it is the smallest and 'second smallest sellers' (the ones with the smallest share except the smallest seller) who mixes over the entire interval. Second smallest sellers have a mass point at the reserve price, and in addition select one of the three familiar possibilities - mix over entire interval, cutoff price or reserve price purely. Note that at least one of the second smallest seller mixes. All other sellers select the reserve price purely. Note that here the smallest seller has larger profit per store than all other sellers, and all sellers but her have equal, lower, profit per store (in all other cases the profit per store is equal for all sellers). The structure is similar to the extended model with two or more smallest firms, where the second smallest sellers behave similarly to the smallest sellers. This change is in line to the simple model suggested by \cite{AsymmetricSearch}, in a model with two different sized sellers.

The Stahl model is dealt extensively in the literature, and is a a very popular model. Numerous extension to the Stahl Model were introduced, and the various extensions are dealing with nearly every aspect of the model. Among those are introducing heterogeneous searchers. Example for such extensions are \cite{Hetrogen} and  \cite{StahlHetro}, where the searchers have different cost for each additional store they visit. They can differ by the search scope, as discussed in  \cite{AsymmetricSearch}, where some stores are near, and thus will be searched first. Another extension introduced advertisement costs, as discussed, for example by \cite{Confusion}. There are also models where already the first price is  costly, such as \cite{FirstCosts}, or no possibility to freely return to previously visited store, such as \cite{NoRecall}. The literature has discussion regarding the sequential search in the model and looks also at non-sequential search, for example in \cite{NonSeqSearch}, or the unknown production cost as shown in \cite{AsymInfo}. Most assumptions of the model introduced by Stahl in \cite{Stahl89} are discussed extensively, except one main assumption, used extensively in the literature. This is the focus on symmetric equilibria, where all sellers select an identical strategy. One of reasons is the mathematical complexity: \cite{CM} and \cite{Rothshild} showed that in symmetric equilibria consumer reserve price must exist, and in asymmetric ones it may not. Reserve price assumption is common in the literature, and therefore, the paper considers only NE with reserve price, yet justifies the rationality behind it. Nevertheless, one should note that additional Equilibria without reserve price may exist, and fall beyond the scope of this paper. 
 
An additional outcome of this model can explain price stickiness, as described for example in \cite{sticky}. Many equilibria found here have mass points on certain prices. This implies that with some probability the price in the previous round can be the same also in the next round, even though the seller is mixing. In reality it is known that that prices do not change too often and are sticky. The results of this model can provide an insight on why it is so, as prices selected with mass points can remain unchanged during several periods.

The structure of the paper is as follows: first the original Stahl model is formally shown and knowledge and structure of the game are discussed. Then I turn to look at asymmetric NE of the original model. Afterward I introduce the extended model and discuss the differences between the models to clarify the nature of the extension. Then I will provide the results to the extended model, firstly in the case with several smallest firms and then with a unique smallest firm. Then the implications of the results are discussed, and suggestions on how those results can be empirically tested.

\section{Model}

The Stahl model, as introduced in \cite{Stahl89} is formally described below. Notation was adjusted to the recent literature on the Stahl model.

There are N sellers, selling an identical good. Each seller owns a \textbf{single} store. The production cost is normalized to 0, and assume that the seller can meet the demand. Additionally, there are buyers, each of whom wishes to buy a unit of the good. The mass of buyers is normalized to 1. This implies that there are many small buyers, each of which is strategically insignificant.

The sellers are identical, and set their price once at the first stage of the game. If the seller mixes then the distribution is selected simultaneously, and only at a later stage the realizations take place.

The buyers are of two types. A fraction $\mu$ of buyers are shoppers, who know where the cheapest price is, and they buy at the cheapest store. In case of a draw they randomize \textbf{uniformly} over all cheapest stores, spreading equally among the cheapest stores. The rest are searchers, who sample prices. Sampling price in the first, randomly and \textbf{uniformly} selected, store is free. It is shown in \cite{FirstCosts} that if it is not the case then some searchers would avoid purchase. If the price there is satisfactory - the searcher will buy there. However, if the price is not satisfactory - the searcher will go on to search, sequentially, in additional stores, where each additional search has a cost $c$. The second (or any later) store is randomly and \textbf{uniformly} selected from the previously unvisited stores, and the searcher may be satisfied, or search further on. When a searcher is satisfied, she has a perfect and free recall. This implies she will buy the item at the cheapest store she had encountered, randomizing \textbf{uniformly} in case of a draw.

There is a developing literature where the sellers are asymmetric, such as \cite{AsymmetricSearch}. The main difference in the extension is that the distribution of searchers among stores is not uniform, but a different one. In this section I concentrate on the case where the sellers are identical, and they can choose different pricing strategies.  This extension is addressed in later sections, where the distribution of searchers among stores would be more generic.

The buyers need to be at both types (namely, $0<\mu<1$). If there are only shoppers - it is the Bertrand competition setting \cite{BertrandEQ}, and if there are only searchers    the Diamond Paradox \cite{Diamond} is encountered, both well studied.

Before going on, make some technical assumptions on the model are introduced, the rationality behind them is explained. The assumptions are as follows:

\begin{itemize} 
\item Throughout this paper it is assumed that the sellers cannot offer a price above some finite bound $M$. This has the interpretation of being the maximal valuation of a buyer for the good.
\item Throughout this paper, it is also assumed that searchers accept any price below $c$. The logic behind it is any price below my further search cost will be accepted, as it is not possible to reduce the cost by searching further.
\item To avoid measure theory problems it is assumed that mixing is possible by setting mass points or by selecting distribution over full measure dense subsets of intervals.
\end{itemize}

\subsection{Reserve Price and Knowledge}\label{AnonKnow}

In the symmetric Stahl model the consumers have a reserve price in NE. The reserve price determines the behavior of consumers - the searcher is satisfied and searches no further if and only if the price is (weakly) below her reserve price, unless all stores are visited. If the price is below the reserve price - the search stops and the consumer purchases the good, if not - the search will continue. If all prices are above the reserve price - the cheapest store will be selected, after searching in all stores. In order to maintain in one line with the vast literature of the model, and being able to compare the results reserve price existence is assumed. However, one needs to specify when and how the reserve price is determined. The reserve price is determined simultaneously to the price strategy choice of the sellers the searchers set a reserve price. The reserve price is identical to all searchers, as was also in the original model. It will be denoted throughout the paper as $P_M$. How the reserved price is determined is dealt with below.

Below is the setting that allows searching, as difference in prices can provide incentives to it. Moreover, it extends the symmetric Stahl model knowledge available to the searchers, as the reserve price is $c$ above the expected price of a seller. There, they knew the mixed strategy chosen by the sellers, and their behavior (whether to search further or not) was adjusted accordingly. Here, as the strategies of the sellers do not have to be identical, a price observed implies something on prices not observed yet. After observing price $p$ in a store, the searcher can estimate the probability that the strategy of the seller is a specific one, and from that induce the expected price in other stores. Therefore, it is important to introduce beliefs and explain how exactly these are adjusted while searching.

The searchers have beliefs regarding the prices set. For each possible (pure and mixed) strategy $s$ of the model is attached a belief, stating how many sellers are actually using this strategy $n(s)$ (clearly the sum of $n(s)$ is n, the number of stores). Each strategy has an expected price, denoted $e(s)$. Now, it is easy to explain how the searcher will determine whether she searches on or not. 

Suppose the searcher observed the price $p$. Let us denote the probability that this price $p$ came from strategy $s$ as $prob(p,s)$. For this the searcher calculates chance that $s$ is selected by some seller and the probability that $p$ is the realization of strategy $s$ (relevant for mixed strategies). One needs to note that if some strategies (with positive $n(s)$) have a mass point on $p$ only those will be considered, and if there are no mass points on $p$ the densities will play a role. Formally:

\begin{equation}
prob(p,s)=\frac{n(s)f(s)}{\sum_{p \in s'}n(s')p(s')}
\end{equation}

Now, if the searcher thinks that strategy $s$ was selected, searching further will yield (in expected terms) the expected price in all the other stores. Therefore, it is the expected price, only that $n(s)$ is now one lower (as $s$ was observed in one of the stores). If $n(s) \leq 1$ $s$ will be simply omitted from further calculations:
\begin{equation}
\frac{\sum_{s': n(s')>0, s' \neq s}n(s')e(s') + (n(s)-1)e(s)}{\sum_s' n(s')}
\end{equation}

\begin{remark}
In the extended model instead of the number of sellers with the relevant strategy the belief will state the number of stores with the relevant strategy.
\end{remark}

Searchers search further only when the expected price in a search is at least $c$ lower than the lowest observed price. Below is an example of how to calculate an expected search price, and additionally illustrates that no reserve price may exist:

\begin{example}
Suppose the search cost $c$ is 0.9 and pricing strategies, equally probable from the beliefs of a searcher, are as follows:
\begin{enumerate}
\item Uniform in [1, 9], Exp. value of 5
\item Uniform in [5, 9], Exp. value of 7
\item Pure strategy of 7.
\end{enumerate}
\end{example}

After observing the price of 7 one is certain with prob. 1 that she had encountered the third strategy seller. An additional search will yield the average between the expected values of the two strategies - namely -  6, making a search worthy. 

After observing the price of $7+\varepsilon$ One knows that she had encountered one of the mixed strategies, and due to a likelihood ratio - twice more probable that it is the second strategy. Therefore, with probability 1/3 it is the first str. and probability 2/3 the second str.

If the first strategy was encountered, then an additional search will end up in ether second or third strategy - both with expected price of 7. 

If it is the second strategy, then an additional search will end up with expected price of 5 or of 7, as both can occur with equal probability (due to the beliefs) expected price in an additional search in this case is 6. 

Combining the two possibilities, when taking into account that the second case is twice more probable, the expected price in an additional search is $(2\cdot6+7)/3=6.333$, making another search not profitable.

Here one sees the problematic assumption of the reserve price - it might be the case that it does not exist. However, in order to maintain in one line with the literature I concentrate on NE with a reserve price. Therefore, when one has a suspected a profile to be a NE one still needs to check whether the searchers there behave rationally, when adopting a reserve price. Later some lemmas will be provided which will help in determining the reserve price condition. Therefore, the set of all possible NE may be wider, as some NE without a reserve price may exist.

\begin{definition}
As the beliefs are on the sellers disregarding their identity, we will refer to this belief setting as 'Anonymous Knowledge'.
\end{definition}

\subsection{Game Structure}

The game is played between the sellers, searchers and the shoppers. The time line of the game is as follows:

\begin{enumerate}
\item Sellers select pricing strategies and consumers set reserve price.
\item Realizations of prices occur for sellers with mixed strategies.
\item Shoppers go and purchase the item at the cheapest store
\item Searchers select a store and observe the price in the store
\item If the price observed is weakly below $P_M$ the searcher is satisfied and purchases the item, if not the search continues
\item All unsatisfied searches select one additional store, pay $c$ and sample the price there.
\item If the price observed is below $P_M$ the searcher is satisfied and purchases the item, if not the search continues
\item ...
\item When the seller observed all stores and observed only prices above $P_M$ she would buy at the cheapest store encountered.
\end{enumerate}

At the time the reserve price and the strategies are determined the knowledge of the various agents of the game is as follows:
\begin{itemize}
\item The sellers are aware of the reserve price set by the searchers
\item The searchers have beliefs about which strategies were actually played by the sellers (see subsection \ref{AnonKnow}).
\item The shoppers will know the real price in each store in the moment it is realized.
\end{itemize}

The probability that seller $i$ sells to the shoppers when offering price $p$ is denoted $\alpha_i(p)$. Let $q$ be defined as the expected quantity that seller $i$ sells when offering price $p$. The expected quantity sold by the seller consists of the expected share of searchers that will purchase at her store, plus the probability she is the cheapest store multiplied by the fraction of shoppers, and is also the market share of the seller (shoppers + searchers mass is normalized to 1).

Note that the reserve price ensures that the searcher will purchase at the last visited store, unless all stores were searched.

The utilities of the game are as follows:
\begin{itemize}
\item The seller utility is the price charged multiplied by the expected quantity sold.
\item The consumer utility is a large constant $M$, from which the price paid for the item and the search costs are subtracted.
\end{itemize}

The NE of the game, under our assumptions, is as follows:
\begin{itemize}
\item Searchers have a reserve price. 
\item The searchers beliefs coincide with the actual strategies played.
\item The reserve price is rational for the searchers, according to 'Anonymous Knowledge'.
\item No seller can unilaterally adjust the pricing strategy and gain profit in expected terms.
\end{itemize}

\begin{remark}
As the sum of the searcher and seller utilities may differ only in the search cost, any strategy profile where the searchers always purchase the item at the first store visited is socially optimal.
\end{remark}

\subsection{Results}

Before stating out the main results of the original model, a number of definitions is required. The reserve price is denoted as $P_M$. Additionally, a specific price denoted as $P_L$, and it is the price solving the following equation:
\begin{eqnarray*}
P_L(\mu+\frac{1-\mu}{n})=P_M\frac{1-\mu}{n}\\
P_L = P_M\frac{1-\mu}{(n-1)\mu+1}
\end{eqnarray*}

If the support of seller $i$ strategy is a positive measure interval from $P_L$ to some price $p_i < P_M$, and in addition mass point at $P_M$, it will be said that seller $i$ has a cutoff price of $p_i$.

Now it is possible to describe the NE of the Stahl model:

\begin{theorem}\label{symmod}
In any NE of the Stahl model with a reserve price there are at most three groups of strategies, as follows:
\begin{enumerate}
\item At least two sellers who have the full support of $[P_L,P_M]$ with some NE dependent continuous full support distr. function $F$.
\item A group of sellers (possibly empty) that select $P_M$ as a pure strategy
\item A group of sellers (possibly empty) with an individual cutoff price, such that below the cutoff price the distribution used is the same $F$ as from the first group.
\end{enumerate}
Additionally, all sellers have the same profit of $P_M(1-\mu)/n$.
\end{theorem}
\begin{remark}
For any combination where the third group is empty and the first group has at least two sellers exists a corresponding NE.
Moreover, the sellers have the same expected profit of $P_M\frac{1-\mu}{n}$ and the searchers buy at the first store they visit. To see this simply adjust the shoppers share to reflect the game when only searchers visiting the mixing sellers exist.
\end{remark}

Illustration of the three types of strategies can be seen on figure \ref{3G}.
\begin{figure}
 \centering
    \includegraphics[height=30mm]{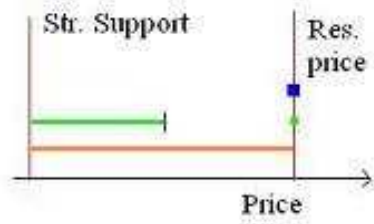}
  \caption{The three types of strategies available in a NE of the extended model}
  \label{3G}
\end{figure}

The proof will be provided in the appendix. However, the first step is required to understand certain results on the extended model. Therefore, it is provided below with a short proof. Several examples will be provided in a later section.

Before continuing I wish to provide some very basic, yet important insights, valid also for the extended model:

\begin{remark}
As noted already in \cite{Stahl89}, due to undercutting no pure NE exist. This is true for the extended model too for the same reasoning.
\end{remark}

\begin{lemma} \label{max_res_lemma}
In both models, no seller offers a price above $P_M$ in NE.
\end{lemma}
Let $p$ be the highest (or supremum) price offered in NE, and $p>P_m$. Such supremum exists as it is assumed that there is a finite bound on the prices.
Let me distinguish between several cases:
\begin{itemize}
\item A unique mass point at $p$ implies profit 0 to the seller offering it. Searchers would go on searching and find something cheaper, whereas shoppers would buy at a cheaper price w.p. 1. A deviation to offer the price $c$ would be a profitable one.
\item No mass points at price $p$ implies profit 0 to all offering it. In case of a supremum price - profit is arbitrarily close to 0. In such case deviation to $c$ is profitable.
\item Some (but not all) offer price $p$ with a mass point. The same case as with a single mass point: the searcher would go on searching until she finds a price cheaper than $p$.
\item All sellers offer $p$ with a mass point - undercutting is profitable. With some positive probability (that all offer price $p$) you would get all the market instead of just $1/n$ of it.
\end{itemize}
To sum it up - for a seller offering a price $p>P_M$ there is a profitable deviation in all cases. \qed

\begin{corollary}
Any NE is socially optimal, whether in the original or the extended model. This is since the total utilities of the sellers and consumers sums up to a constant, as long as the searchers buy at the first store they visit.
\end{corollary}

I now show a lemma which will assist in determining the reserve price condition for the searchers:

\begin{lemma}\label{SearcherCondLemma}
Suppose that in a NE every seller has the expected price of at least $P_M-c$. Then setting $P_M$ as a reserve price is rational for the searchers.
\end{lemma}

It is not possible to observe a price above $P_M$, therefore, the searcher always stops searching after the first store visited. It is still required to show that after the first price observed it is not rational for the searcher to continue searching.

Suppose a price $q$ was observed. As $q \leq P_M$ it is required to show that an expected price in a search is at least $q-c$. As the expected price in a search is a convex combination of some of the expected prices of sellers it is larger than a lower bound on such expected values. The lower bound on these expected values is $P_M-c$. Therefore, the expected price obtained in an additional store is at least $P_M-c>q-c$, making an additional search unprofitable. \qed

Note that the condition here is only a sufficient one, and it might be the case that additional reserve prices may be rational for searchers. Therefore, the asymmetric NE found here may do not cover all the possible NE of the model.

\subsection{Equilibrium Distribution}\label{EQ_distr}

Now it is possible to elaborate on the structure of the $F$ function which is used in equilibrium by sellers, and what reserve price can be used. Suppose that in equilibrium we have $O=\{1, 2, \ldots o\}$ sellers with orange strategy (mixing over entire support), $B$ sellers with blue strategy (pure reserve price) and $G= \{1,2, \ldots g\}$ sellers with green strategy (cutoff price strategy), with the cutoff prices of $cp_1, cp_2, \ldots cp_g$ and mass points at the reserve price are with mass of $a_1, a_2, \ldots a_g$.

Let us denote the set of sellers with cutoff point below some price $p$ as $B(p)$.

From the structure of the equilibrium all sellers have equal profit. Additionally, all sellers have $P_M$ in support and the reserve price attracts no shoppers. Therefore, the profit for all sellers is:
\begin{equation}
\pi = P_M(1-\mu)/n
\end{equation}

For any price $p$ the expected profit needs to be equal to the expression above. At price $p$ seller $i$ has a certain probability $\alpha_i(p)$ to attract shoppers, if she is the cheapest. This can be calculated as follows:
\begin{itemize}
\item For each seller $j \neq i$, calculate the probability that $j$ offers a price above $p$
\item Multiply these probabilities
\end{itemize}

Let $p$ be a price in $(P_L,P_M)$. For group $O$ this probability is clear and equal to $1-F(p)$. For group $B$ - it is zero. For group $G$ we need to distinguish between two cases: ether $p \in B(p)$ and the probability is $1-F(p)$, or $p\not \in B(p)$ and then it is equal $a(p)$. Combining the cases we get that the expression for the expected profit is as follows:
\begin{equation}
\pi = P_M(1-\mu)/N = p[(1-\mu)/n+\mu(\prod_{j \in O \cup B(p)} (1-F(p)) \prod_{j \in G \setminus B(p)} (a_j))]
\end{equation}
As the $F$ function is the same we can simplify and get:
\begin{equation}
p[(1-\mu)/n+\mu((1-F(p))^{o+|B(p)|} \prod_{j \in G \setminus B(p)} (1-a_j))]=P_M(1-\mu)/n
\end{equation}
Extracting $F(p)$ form this equation will yield:
\begin{equation}
F(p)=\sqrt[o+|B(p)|]{1-(\frac{P_M}{p}-1)\frac{1-\mu}{n \prod_{j \in G \setminus B(p)} (a_j)}}
\end{equation}

Note that at the point of the cutoff price $a_j(p)=1-F(p)$, and therefore $F$ will be continuous, and as a certain expression instead of decreasing remains constant will also be differentiable. Therefore, it is still possible to calculate the density and expected value regularly. The last step, based on lemma \ref{SearcherCondLemma}, require finding the expected value $E(F)$, and setting the reserve price at $E(F)+v$. As this step is technical and the expressions involved are complex, this step is not done here for the general case. At the section with examples specific cases are provided.

\section{Extended Model}

In order to deal with heterogeneous firms an extension to the model is required. Instead of a single store, seller may have several stores, which can vary among sellers. This makes searching uniformly random among stores and it is no longer uniformly random among sellers.

The extension of the model is very simple. The searchers do not spread uniformly over the stores, but according to some given propensities. Seller $i$ has $n_i$ stores of the available N stores, and the searchers are distributed according to this number. Note that it would be more convenient to think in the extended model of sellers as store chains, with various number of stores.

This makes the fraction of searchers initially visiting seller $i$ equal to $\frac{n_i}{N}$ instead of $1/n$. Further search is also done according to the propensities, and the probability to visit seller $i$ is $n_i$ divided by the sum of the $n_i$'s of previously unvisited stores.

The section describing the model remains the same, except the following points: 'uniform' should be changed to 'according to the propensities $n_i$'. An additional difference is that after visiting and being unsatisfied with a single store of a given seller the searcher would not return to visit another store of the same seller. The values of $n_i$ are common knowledge among sellers. Once a store is visited the size of the store is revealed to the searcher. Searchers beliefs are as follows: for any possible strategy $s$ the searchers have belief on how many stores applied this pricing strategy (see also section \ref{AnonKnow}. Store number of the sellers remains unknown to searchers.

One technical assumption I make trough this section is that the smallest firm (the one with the smallest number of stores) is not unique. The case with a unique smallest firm will be dealt in the next section.

Let the smallest value of the size parameter be denoted as $n_m$.

\subsection{Difference between Models}

The main point of the extension is to catch the fact depicted in figure \ref{DiscountersNum}. Heterogeneous sellers have different number of stores, and due to advertising keep the price fixed in all of the stores of a given seller. This is not always the case, as some goods do not have to be fixed over all stores of a seller. Therefore, one needs a distinction between \textbf{store goods}, that each store set the prices individually, and \textbf{chain goods} which have a fixed price in all chains stores.

A good example for the distinction can be understood via two examples for a seller - a bank and a fuel station. Typically all offers of a given bank do not vary among branches. One would get the same offers for a mortgage, credit, interest rate for deposits and other banking products no matter the specific branch of the bank one approaches. Clearly, a different bank would make different offers, but usually it is the case that a specific branch of a bank 'X' on street 'A' would have the same offers as the branch of bank 'X' in street 'B'. When one looks at fuel stations, one sees the exact opposite. Every single fuel station offers station-specific prices, and it is usually the case that fuel station of firm 'Y' on street 'A' would have a different price than a fuel station of the same firm 'Y' on street 'B'. 

For the fuel station example, the original Stahl model would suffice, as there being unsatisfied with a specific station does not imply avoiding the seller completely. However, if after visiting one branch of bank 'X' one does not find a satisfactory offer, there is no reason to visit yet another branch of the same bank. Therefore, present here are both aspects of the extension - the probability to encounter each seller is proportional to the number of stores the seller has and the fact that at most a single store of a given seller would be visited by a given searcher.

\begin{remark}
To emphasize the difference I will refer to the models as '\textbf{original}' and '\textbf{extended}'.
\end{remark}

\subsection{Results}

First, a proof for the fact that no symmetric NE exists in the extended model is provided. This result is valid also for the next section, when smallest seller is unique.

\begin{lemma} \label{no_asym_extended}
In the extended model no symmetric NE (where all sellers choose the same pricing str.) exists.
\end{lemma}

As before I concentrate only on NE with a reserve price, denoted $P_M$. Due to undercutting no pure NE can exist. Suppose that exists a mixed strategy symmetric NE. Let the pricing strategy in the symmetric NE be denoted as $s$, and let $F$ be the distribution function describing the pricing strategy $s$.

As noted before, no seller sets a price above $P_M$, as in such case the highest price in the support yields profit 0.

Suppose that two prices $p$ and $q$ are in the support of $s$ with positive density or mass point. Additionally, exist two sellers $i,j$ such that $n_i>n_j$.

Note that due to symmetry of the strategy choice, the probability of a seller to attract shoppers with price $p$ is equal to $(1-F(p))^{n-1}$. The next step is to write down the profits of seller $i$, for both prices. Those need to be equal, as mixing is possible only between prices that yield the same expected profit:
\begin{equation}
\pi_i(p)= p[(1-F(p))^{n-1} \mu + (1-\mu)n_i/N] =  q[(1-F(q))^{n-1} \mu + n_i] = \pi_i(q)
\end{equation}
This contains the expected quantity sold by the seller: prob. to be cheapest times the quantity of shoppers and the quantity of the searchers, which contain  only the initial searchers visiting the store (due to lemma \ref{max_res_lemma}).

Similarly, the profit of seller $j$ is as follows:
\begin{equation}
\pi_j(p)= p[(1-F(p))^{n-1} \mu + (1-\mu)n_j/N] =  q[(1-F(q))^{n-1} \mu + (1-\mu)n_j/N] = \pi_j(q)
\end{equation}
After subtracting the second equation from the first one obtains the following:
\begin{equation}
p[ n_i - n_j] =  q[n_i - n_j]
\end{equation}
From here, ether $p=q$ or $n_i = n_j$, both cannot occur due to our assumptions. \qed

Then one may ask are there any NE, and if so, how do they look like. The following theorem provides an answer to that question:

\begin{theorem} \label{asymmod}
In the extended Stahl Model, with at least two firms, the asymmetric NE with reserve price $P_M$ must have the following form:
\begin{itemize}
\item All the sellers who have a larger number of stores than $n_m$ select the reserve price as a pure strategy. 
\item The agents with $n_m$ need to choose their strategies according to theorem \ref{symmod}, where all shoppers and only the searchers visiting initially one of the smallest firms take part.
\item Profit of seller $i$ is $Const \cdot n_i$.
\item Searchers buy at the first store they visit
\end{itemize}
\end{theorem}

In any NE the searchers visiting one of the larger sellers will observe the reserve price. The sellers with the smallest chains, the shoppers and the remaining searchers will then play the original Stahl model. An equilibrium of this game will determine the prices offered by the smallest sellers.

Illustrating the various consumers of the game is depicted on figure \ref{consvar}.
\begin{figure}
 \centering
    \includegraphics[height=30mm]{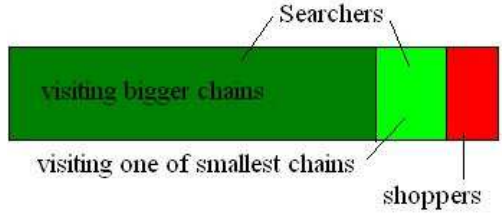}
  \caption{Consumers and their initial visit}
  \label{consvar}
\end{figure}

As the distribution function involves only firms with equal size, the calculations required are identical to the ones done for the original model in subsection \ref{EQ_distr}.

\section{Unique Smallest Firm}

So far in the extended model the case where the smallest firm was unique was omitted. This case is dealt here. In this structure exist NE, however these slightly differ in structure from the previous one. In \cite{AsymmetricSearch} Stahl model with two different sized sellers is discussed, but in a more general setting and more sellers. The NE below stands in line with two sellers behave according to the results in \cite{AsymmetricSearch}, while the rest offer the reserve price purely.

Let the smallest seller be denoted as $m$ and (one of) the second smallest seller as $j$, with corresponding store numbers $n_m$ and $n_j$. Additionally the shares of consumers that would visit seller $i$ as searchers (note lemma \ref{max_res_lemma}) are denoted as follows:

\begin{equation}
Src_i = \frac{n_i}{N}(1-\mu) 
\end{equation}

\begin{proposition}\label{singlefirm_ex}
Exists a NE with a reserve price $P_M$ such that:
\begin{enumerate}
\item All sellers except $m$ and $j$ select $P_M$ purely.
\item The lowest price in support is $P_L = P_M(\frac{Src_j}{\mu+Src_j})$.
\item $m$ and $j$ mix on the entire interval $(P_L,P_M)$.
\item The distributions are as provided below, and $j$ has a mass point at $P_M$.
\end{enumerate}
\begin{eqnarray}
F_m(p)=1-\frac{Src_j}{\mu}(\frac{P_M}{p}-1)\\
F_j(p)=(1-\frac{P_L}{p})(1+\frac{Src_m}{\mu})\\
P_M=\frac{c}{1 - ln(\frac{Src_j+\mu}{Src_j})\cdot \frac{Src_j}{\mu}}
\end{eqnarray}
\end{proposition}

\begin{proof}

Note that from the construction of $P_M$ it is clear that it is above $c$.

First see that there is no deviation to $m$ and $j$:

It is easy to verify that there is a constant profit for sellers $m$ and $j$ in the interval, since:
\begin{eqnarray}
\pi_m(p)=((1-F_j(p))\mu+Src_m)\\
\pi_j(p)=((1-F_m(p))\mu+Src_j)
\end{eqnarray}
Offering prices below $P_L$ is not profitable, as already in $P_L$ one attracts the shoppers w.p.1. Prices above $P_M$ will not be offered due to lemma \ref{max_res_lemma}. Therefore, $m$ and $j$ have no profitable deviation.

Seller $k$ who is not $m$ or $j$ would similarly refrain form selecting prices below $P_L$ or above $P_M$. Deviating to a price $p \in (P_L, P_M)$ would yield the following profit:
\begin{eqnarray*}
\pi_k(p)=p((1-F_m(p))(1-F_j(p))\mu+Src_k) < \\
p((1-F_m(p))\mu)+p(Src_k)
\end{eqnarray*}
Seller $j$ has price $p$ in support and therefore:
\begin{equation}
\pi_j = p((1-F_m(p))\mu)=(P_M-p)Src_j
\end{equation}
Combining the equations, the following expression is obtained:
\begin{equation}
\pi_k(p)<(P_m-p)Src_j+pSrc_k
\end{equation}
Note that the size of $k$ is at least the size of $j$, implying that $Src_k \geq Src_j$. Therefore, the profit of seller $k$ when offering price $p \in (P_L,P_M)$ is below $P_M Src_k$. However, this profit is obtained by $k$ when offering $P_M$, and therefore, has no profitable deviation from $P_M$.

Lastly, the reserve price is rational. Note that if one compares the derivatives of $F_m$ and $F_j$:
\begin{eqnarray}
f_m=\frac{Src_j P_m}{\mu p^2}\\
f_j=\frac{(\mu+Src_m)P_L}{\mu p^2}
\end{eqnarray}
Using the facts that $Src_m<Src_j$ and $P_M Src_j = P_L (\mu+Src_j)$ it is easy to see that $f_m(p)>f_j(p)$ for any price in $(P_L,P_M)$. As the distribution of $j$ has a mass point at the maximal price of $P_M$, the expected value of $F_m$ is smaller than the one of $F_j$. From lemma \ref{SearcherCondLemma} in order for the reserve price to be rational $E(F_m)$ needs to be at least $P_M-c$, which occurs with equality, due to the structure of $P_M$:
\begin{equation}
E(F_m(p))=\int_{P_L}{P_M}p f_m(p) = \frac{Src_j P_m}{\mu}\int_{P_L}{P_M}\frac{1}{p} = \frac{Src_j P_m}{\mu}\cdot ln\left(\frac{Src_j+\mu}{Src_j}\right) = P_m-c
\end{equation}
\end{proof}

Note that here all sellers but $m$ have the same profit per store and a mass point at $P_M$. The smallest firm has a larger profit per store, and offers more generous discounts. Also, additional equilibria may exist where several of the smallest firms after $m$ also mix, but again, with smaller mass than $m$. This comes in line with the results pointed out in \cite{AsymmetricSearch}. The next step is to denote the general structure of NE in such case. It is given in the theorem below:

\begin{theorem}\label{singlefirm}
In the case of a unique smallest seller the NE with a reserve price $P_M$ of the game look as follows:
\begin{itemize}
\item All sellers with size above $n_j$ select the reserve price purely.
\item The lowest price in the support union is $P_L = P_M(\frac{Src_j}{\mu+Src_j})$.
\item Seller $m$ mixes with a continuous, dense distr. function $f_m$ over $(P_L,P_M)$.
\item Some sellers with size $n_j$ also mix over the entire interval with a continuous dense $F_j$, such that $F_j(p)<F_m(p)$ for all $p \in (P_L,P_M)$, and in addition have a mass point at $P_M$. 
\item As in the previous case some sellers with size $n_j$ can also select $P_M$ purely or have a cutoff price.
 The sellers with a cutoff price use the same $F_j$ below the cutoff price.
\item All sellers except $m$ have the same profit per store, and $m$ has a higher profit per store.
\end{itemize}
\end{theorem}

The calculation of the equilibrium distribution is done in similar lines to the proposition \ref{singlefirm_ex}

\section{Examples}

The theorems described how to look for such NE, but without examples it might be harder to perceive. Here I provide two examples for Asymmetric NE, one for the original Stahl Model and one for the extended Model. The examples suggest a wide class of NE, where some sellers select the reserve price purely, whereas all other sellers select the symmetric NE strategy of the model with the remaining number of searchers. 

\subsection{Original Model Example}

Consider the Stahl model with 3 sellers and a shoppers fraction of 1/4.

The following asymmetric NE exists:
\begin{itemize}
\item The searchers have a reserve price of $P_M=c/(1-ln2)>c$
\item One of the sellers offers the reserve price as a pure strategy.
\item The other two sellers use an atomless distribution function $F(p)=2-P_M/p$ on $[P_M/2,P_M]$
\end{itemize}

Note that $1/4$ is the mass of searchers visiting each of the stores initially.

The pure str. agent receives the profit of $P_M/4$.

Suppose the mixed str. agent selects a price $p \in [P_L,P_M)$. Then, her expected profit would be:
\begin{equation}
p(\frac{1-F(p)}{4}+\frac{1}{4}) = \frac{p}{4}(2-F(p)) = \frac{p}{4}\frac{P_M}{p} = \frac{P_M}{4}
\end{equation}
Clearly, if the pure str. agent selects a price in $(P_L,P_M)$ her prob. to sell to shoppers is $(1-F(p))^2 <(1-F(p)$, and therefore, such deviation is not profitable. Similarly, selecting $P_L$ would lead to the same profit as selecting $P_M$.

Any agent selecting prices above $P_M$ would not sell to anyone, and selecting a price below $P_L$ yields less profit.

One last thing to check is the searcher condition. Sufficient for this would be to check that the expected price of the mixed str. seller is at least $P_M-c$. 

The density function, which is the derivative of the distribution function, is $P_M/p^2$. Therefore, the expected value is:
\begin{eqnarray}
E = \int_{P_M/2}^{P_M} pf(p) = \int_{P_M/2}^{P_M} P_M/p = P_M(ln(P_M/P_L)) = P_M(ln2)
\end{eqnarray}

Thus, the expected price of a mixing seller, $E=P_M ln2$. Since $P_M(1-ln2)=c$, it is easy to see that $P_M - E = c$, or $P_M-c=E$ as required.

If a searcher did not observe the price of $P_M$ but a lower one, she know that she had encountered a mixed price agent. Additional search will yield with prob. 0.5 another mixed agent with expected price of $P_M-c$, or prob. 0.5 of a pure agent and price $P_M$. Combined - expected price in an additional search is $P_M-c/2$, making the additional search not profitable after observing a price below $P_M$, due to the search price $c$.

If a searcher observed a price of $P_M$ she know that she encountered a pure str. seller, and if she searches further she will get the expected price of $P_M-c$. Here the searcher is indifferent whether to search on or not. Therefore, it is an equilibrium.




\subsection{Extended Model Example}

I take a similar example to the symmetric model. Again, with 3 sellers, but now only 1/6 are shoppers. 1/2 of the consumers are searchers initially visiting one of the stores, and 1/6 of the consumers are searchers initially visiting each of the two others. The corresponding number of stores is, for example, $3,1,1$.

The following asymmetric NE exists:
\begin{itemize}
\item The searchers have a reserve price of $P_M=c/(1-ln2)>c$
\item The seller with the larger store number offers the reserve price as a pure strategy.
\item The other two sellers use the same atomless distribution function on $[P_M/2,P_M]$
\item The distribution function for the two mixing sellers is $F(p)=2-P_M/p$.
\end{itemize}

The searchers condition is analogous to the symmetric model, and therefore, the reserve price would be indeed rational.

One needs to check that no seller wishes to deviate. Firstly, note that prices of above $P_M$ or below $P_L=P_M/2$ are not profitable for all sellers. Already at $P_L$ there is a prob. 1 to sell to shoppers and there is no need in a further discount. Prices above $P_M$ would cause a quantity sold to be 0.

The profit for the pure str. seller is $P_M/2$ and for the other two is $P_M/6$ when offering the price $P_M$.

The profit of the mixed str. seller when she offers a price $p \in [P_L,P_M)$, is as follows:
\begin{equation}
p(\frac{1}{6}+\frac{1-F(p)}{6}) = \frac{p}{6}(2-2+\frac{P_M}{p}) = \frac{P_M}{6}
\end{equation}
Therefore, the mixing agents are indeed indifferent between the prices in the interval.

Lastly, one needs to show that the pure str. agent would not deviate to a lower price. His profit when offering a lower price $p$ is:
\begin{equation}
p(\frac{1}{2}+\frac{(1-F(p))^2}{6}) \leq p(\frac{1}{2}+\frac{1-F(p)}{6}) < p(\frac{1}{2}+\frac{1-F(p)}{2}) = \frac{P_M}{2}
\end{equation}
The first inequality due to $F(p)$ being between 0 and 1. The second strict inequality due to $1-F(p)$ being strictly positive when $p<P_M$. The third equality is algebraic.

It is clearly visible that the pure str. agent has no incentive to deviate to a different str., and therefore it is a NE.

An interesting point to note is that the average seller profit in this example is higher than in the previous one. Here, the three sellers together receive an expected profit of $(1/2+1/3+1/3)P_M=5/6P_M$ whereas in the previous example they got only $3/4P_M$. The reason behind it is simple: more searchers (half of them) bought at the reservation price in the asymmetric case, in comparison to 1/4 of them in the symmetric case. Clearly, if all but the smallest sellers offer the high price, the fraction of buyers buying at that price will be higher.

\subsection{Additional Examples}
It is possible to construct additional examples for the original model as follows: Add to a symmetric NE setting an additional seller that charges purely the reserve price. Then, by adjusting the searchers fractions, similarly to the example provided above a NE will be obtained. For the extended model exists a NE where the sellers with the lowest store number ignore the searchers visiting one of the larger sellers and obtain among them a symmetric (for example) NE, and the other sellers set a pure strategy of $P_M$. The way to show that the profiles are NE are similar to the way that the examples above were shown, for example, by having more than one agent selecting the reserve price as a pure strategy, or a seller having a cutoff price.

\section{Discussion}

Here is a short discussion over the models and result. first, the structure and economic motivation on the results is provided. Afterward, a couple of situations are shown where importance of the extension becomes clear. Then, a couple of empiric tests are suggested to verify whether the results hold in the lab. Lastly, some points to future research are suggested.

\subsection{NE Structure}

The structure of the NE is in one line with the results in \cite{AsymmetricSearch} and \cite{Monopolist}. There, a single firm has a larger size, and this firm will charge the reserve price purely. 'However, this paper suggests a more wide setting and more possibilities for firm sizes. In this extension we see that indeed larger firms find no incentive to compete on shoppers. The reason behind it can be seen easily form the profit structure. The profit of seller $i$ is as follows:
\begin{equation}
\pi_i(p)=p\alpha_i(p)+p(1-\mu)n_i/N
\end{equation}

The profit consists of two components - expected profit form shoppers and profit from searchers. Setting a lower price has two effects - on one side it increases the probability to attract shoppers, but on the other it reduces the profit from the searchers. Note that the first, positive, effect is more size independent (probability to be cheapest increases similarly no matter your size), whereas the second effect is size defendant and is more significant for larger sellers. Therefore, a larger seller will find it less attractive to offer discounts. In the case of at least two smallest sellers, these will compete one with the other, and larger sellers will not even bother to enter the 'shoppers market', by sticking to the reserve price. In the case of a unique smallest seller, the competition will be with some 'second smallest' sellers, and sellers above that 'second smallest' size will stay out.

The three types of strategies in the NE have some economic motivation. The mixing seller wishes to compete over the shoppers when the pure reserve price seller does not to bother with the shoppers. Those kind of behavior are common in the economic world, and not in all cases all will compete as predicted by the symmetric NE. If only a single seller decides to compete, she will have monopolistic profits, which would attract additional competitors, and therefore, in NE at least two sellers will compete for the shoppers. As suggested in \cite{AsymmetricSearch}, in the case of a smallest unique seller (their setting is of two seller with different size), the smaller seller will offer lower prices with higher probability.

The cutoff price is for sellers that do not wish to be bothered with small probabilities. There are several effects that may cause a seller to refrain from sufficiently low probability events, for example see \cite{RareEvent}. Then, such seller will compete for shoppers, but only at prices that yield the benefit of getting the shoppers from high enough probability. When the probability to attract shoppers is lower than this individual threshold, the seller prefers to refrain from the shoppers market and select the reserve price purely.

These three strategies in addition to the size implication explain the behavior of sellers in this model.

\subsection{Examples}
First example uses the data from Table 2 in from \cite{HungMilk}. This paper discusses the  pricing of a homogeneous good (milk) in 8 discounters in Hungary over a period of 5 years (2004-2008). The stores number data was not available in the paper, and current stores number was obtained (Jan 2012), with two exceptions. One of the supermarket chains, PLUS, was purchased by another, InterSpar, after the relevant period. Therefore, from the current number of stores by InterSpar the number of PLUS stores (~170 stores) was subtracted. The current number of stores serves as an indicator to the number of stores in the research period, and is divided into several distinct groups by size. This provides enough insight, and the idea that smaller chains are usually cheaper. Note that there can be additional factors (such as location of stores) affecting the price, however, one factor can indeed be the chain size. The table is as follows:

\begin{tabular}{c c c}
Chain Name & Stores Number in Hungary & Avg. price of milk \\
InterSpar	& 50-100\footnote{Excluding later bought Plus Stores} & 182 HUF \\
Cora	& Below 20 & 198 HUF \\
Match	& Below 20 & 200 HUF \\
Tesco	& 200-400 & 205 HUF \\
Auchan & Below 20	& 211 HUF \\
CBA	& Above 500	& 213 HUF \\
Plus\footnote{Got Bankrupt and bought by InterSpar} & 100-200 & 230 HUF \\
COOP & Above 500 & 240 HUF 
\end{tabular}

Another example is from drogerie stores in Germany. GKL\footnote{\tiny{Conducted in Jan. 2012, comparing prices in 1700 stores in Germany}} research found that Schlecker is 10-20\% more expensive than competitors. Additionaly, the number of stores in January 2012 of the various chains in Germany and Europe is given in figure \ref{schlecker}. Again we see a tendency that the largest chain is the more expensive one. Again, additional factors may have an effect here, but it seems that chain size plays a role when determining prices.

\begin{figure}
 \centering
    \includegraphics[width=60mm]{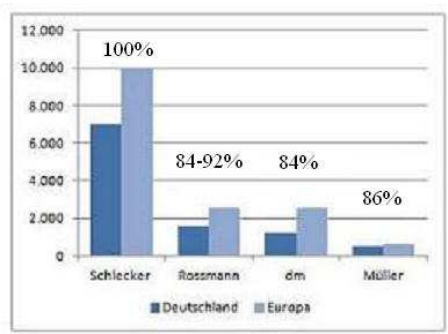}
  \caption{Consumers and their initial visit}
  \label{schlecker}
\end{figure}

To conclude, it seems that there should be some positive correlation between chain size and price. An empiric research checking this connection explicitly will be able to determine how strong it is.

\subsection{Empiric tests and Policy Suggestions}

The structure of NE allow to run several empiric tests on a database containing pricing and chain size data. Firstly, sellers may play an asymmetric NE, and the results here suggest some differences from the classical Stahl Model. Firstly, there will be a higher probability for reserve price. In any asymmetric NE some sellers select the reserve price with a mass point. This implies that the reserve price will be more commonly selected. Similarly, larger discounts will be more rare, as the reserve price will be more common. When sellers with different size are examined, one should see correlation between chain size and price. Moreover, in a dynamic setting there should be some price sticking, as mass points exist. 

When examining the probability for a consumer to encounter a lower price, it is clearly visible that the larger is the variance in store sizes, and the rarer the smallest stores are, the closer expected price paid is to the reserve price. This is an additional factor to examine, and it suggests an interesting policy decision. If the regulator wishes to reduce goods prices, reducing the variance among the selling firms can reduce the price, as exists a NE where less sellers select the reserve price purely.

\subsection{Future Research}

The results here open several important questions, which leave place for a fruitful future research. Firstly, the assumption here is that a reserve price exists. There may be additional NE without a reserve price, and an interesting question is whether such exist and how do these look like. This will allow to fully characterize all NE of the model and fully explain behavior of sellers. An additional question is combined with the determination of the reserve price. What is the full set of reserve prices under a certain setting, as here only a lemma provides a sufficient condition for the rationality of it. Moreover, which reserve price will the consumers set in order to minimize their price. On the other hand - with which NE should the sellers respond. What is the best NE for sellers and what is the best NE for consumer, will sellers prefer to mix, or to have a specific cutoff price? This question of consumer welfare and seller welfare will provide an important insight on behavior of these groups, and can provide a policy decision for a regulator in order to set the price lower or higher.

The Stahl model is a very important tool and the model is being used and applied in numerous papers. I hope that this paper provides an additional important insight which will make the Stahl model more applicable and more realistic. Additionally, any of the further research topics suggested here will provide yet another important block to the model, and to explaining behavior of consumers and sellers.

\begin{appendix}

\section{Symmetric Model}

Here I show the proof to theorem \ref{symmod}. This is shown in a sequence of lemmas, first dealing with the regular Stahl model and then dealing with the extended model.

\subsection{Mass Points and Highest offered Price}

\begin{lemma}
There are no mass points at any price that can attract shoppers with positive prob.
\end{lemma}
If at price $q$ there is a mass point by a single seller $i$, price just above it is strictly less profitable for all others, and therefore would not be selected, as there the chance to attract shoppers drops discontinuously. Thus, seller $i$ can set the mass point higher and gain more profit. In the case of mass points by several sellers at price $p$ undercutting is possible, which probability to attract shoppers discontinuously. Therefore, there are no mass points at prices that can attract shoppers. \qed

\begin{lemma}
All sellers select $P_M$ as the supremum point of their strategy support.
\end{lemma}

From lemma \ref{max_res_lemma} it cannot be higher than $P_M$.

Suppose that the supremum price of seller $i$ is $p<P_M$. For any price above $p$ and below $P_M$ the probability to sell to shoppers is 0. Therefore, in equilibrium no seller would select a price in $(p,P_M)$. Additionally, suppose that seller $i$ has the lowest support supremum.

All sellers cannot have a mass point at $p$, as in such case undercutting would be profitable. From previous lemma seller $i$ has no mass point at price $p$. Thus, probability to sell to shoppers at price $p$ is 0, for all other shoppers, and no other seller would have a mass point at this price. Therefore, a deviation exists to seller $i$, where $i$ selects prices arbitrarily close to $P_M$ instead of prices arbitrarily close to $p$ is profitable. \qed

\begin{remark}
Note that this implies equal profit to all sellers in any equilibrium, or all but one have equal and one higher. 

If at least two sellers do not have a mass point at $P_M$ the probability that shoppers buy at $P_M$ is 0. Moreover, if only one seller has no mass point at $P_M$ she has weakly higher profit than all other sellers.
\end{remark}

The two lemmas combined imply that there can be no mass points at any price except for $P_M$.

\subsection{Single Interval and Profit Equivalence}

\begin{definition}
Let $\alpha_i(p)$ be denoted as the probability that $p$ is the cheapest price, if seller $i$ selects it. Explicitly: what is the probability of seller $i$ to sell to shoppers given she selects price $p$. As the distribution is atomless except (maybe) $P_M$, one can define $\alpha_i(p)$ as the product of 'Probability that seller $j$ sets price above p', which is denoted as $\beta_j(p)$. 
Formally:
\begin{eqnarray}
\beta_j(p) = 1-F_j(p)\\
\alpha_j(p) = \prod_{j\neq i} \beta_j(p)
\end{eqnarray}

\end{definition}

\begin{lemma}
Exists an interval $I$ such that the union of the seller strategies is contained in $I$ and dense in it.
\end{lemma}

Suppose exists an interval $[a,b]$ ($a<b<P_M$) such that sellers select prices only below $a$ and above $b$, and exist prices both below $a$ and above $b$. Let $p-$ be the highest price below $a$ that is in the support union of the sellers. A seller can deviate from $p-$ and prices just below it to $b$, and sellers arbitrarily close to all of her previous quantity: 

The searchers behavior does not change, as the prices are below $P_M$. Since the probability for someone to select a price just below $p$ is arbitrarily small, the decrease in probability to sell to shoppers is arbitrarily small. 

The profit form raising the price is much higher than such arbitrarily small loss, as it is at least $(b-p)(1-\mu)/n$, as the searchers pay strictly more after the deviation. Therefore, if the support is not continuous there is a profitable deviation. \qed

\begin{corollary}
Exists an interval $I=[P_L,P_M]$, such that any NE strategy profile the sellers randomize continuously over $I$, and possibly some sellers set mass points at $P_M$.
\end{corollary}

\begin{lemma}
The previous lemma holds also for two sellers. Meaning - any interval has a non empty intersection with the support of at least two sellers.
\end{lemma}
Suppose that all points in an interval $[p,p']$ ($p<'p<P_M$) are selected at most by one seller. Additionally, from previous lemma this seller needs to have in support the entire interval. Than exists a profitable deviation for her would be to set a mass point at $p'$ instead of selecting the original distribution over the interval.
\qed

\begin{corollary}
Any interval between $P_L$ and $P_M$ has points in the support of at least two sellers.
\end{corollary}

\begin{lemma}
All sellers have the same profit. 
\end{lemma}
The only case that needed to be shown is as follows: If $n-1$ sellers have the same profit, the other seller cannot have a profit above them. It was shown before that if at least two sellers do not have mass points at $P_M$ all sellers have equal profit. If only one seller has no mass point at $P_M$ then she must have a higher profit. This is since she can always deviate to a pure strategy offering $P_M$.

Suppose seller $i$ is the only seller who does not offer a mass point at $P_M$.
Let $p_i$ be the lowest (infimum if needed) price in the support of $i$. As it has a higher profit than all other players this price cannot be the lowest price in the support union. Note that due to previous lemmas seller $i$ sets no mass point at $p_i$, $F_i(p_i)=0$. If no seller selects a price below $p_i$ then it is not possible for seller $i$ to have a higher profit than other sellers, as other sellers could get the same profit as $i$ gets with $p_i$. Denote a seller $j \neq i$, and examine the profits of seller $i$ and $j$. As noted before, $\pi_i > \pi(j)$.

The profit of seller $i$ offering $p_i$ is (remember all searchers visit exactly one store):
\begin{equation}
\pi_i(p_i)=p_i((1-F_j(p)\prod_{k \neq i,j}(1-F_k(p_i))\mu+(1-\mu)/n)
\end{equation}
The profit of seller $j$:
\begin{equation}
\pi_j(p_i)=p_i((1-F_i(p)\prod_{k \neq i,j}(1-F_k(p_i))\mu+(1-\mu)/n)
\end{equation}
Since $0=F_i(p_i) \leq F_j(p_i)$ the profit of $j$ when offering $p_i$ is weakly higher than the profit of $i$ when offering $p_i$. This contradicts the fact that seller $i$ must have a higher profit than seller $j$. \qed

I have shown that the support union is equal to some interval $I$. Let $P_L$ be the lowest price in this interval. As the mixing sellers need to be indifferent between all the strategies they mix one can say that:
\begin{equation}
P_L = \frac{(1-\mu)/n}{\mu+(1-\mu)/n}P_M
\end{equation}
Clearly, the searchers do not search at this price, as $P_L$ is the cheapest price that can exist in EQ. Additionally, when a seller selects this price she is certain to sell to shoppers.

\subsection{Symmetry}

In this subsection I will discuss the symmetries in the NE distribution functions, and see where they can differ.

\begin{lemma}
The following inequality needs to be satisfied for any $p \in (P_L,P_M)$:
\begin{equation}\label{symNEineq}
p(\alpha_i(p)\mu+ \frac{1-\mu}{n})\leq P_M(\frac{1-\mu}{n})
\end{equation}\end{lemma}

If it is strictly larger than price $p$ will be more profitable than $P_M$, which due to previous lemmas cannot occur in NE. Moreover, if seller $i$ selects price $p$ there must be an equality, as the profit $i$ gets from any price she selects has to be equal to $P_M(\frac{1-\mu}{n})$.
\qed

The following observation will be crucial in understanding the asymmetric NE:
\begin{corollary}
Only the seller(s) with the maximal $\alpha$ among the sellers may select the corresponding price.
\end{corollary}

Note that since there are no mass points $\alpha$ and $\beta$ change continuously, except possibly at $P_M$. Note that at $P_M$ the $\alpha$ of each seller approaches 0 continuously as $P_M$ is approached, as the probability to sell to shoppers with price $P_M$ is 0. Adding the fact that there are no mass points below $P_M$, it is clear that $\alpha$ a continuous function.

\begin{lemma}
Let $I$ be an open interval in $[P_L,P_M]$. Suppose that seller $i$ selects a price in $I$ with some positive probability. Let $j$ be a different seller. Then, also seller $j$ must select a price from $I$ with same positive probability, or not to select any prices in or above the interval $I$, except $P_M$.
\end{lemma}

From previous lemma it is known that each interval is selected by at least two sellers, and done so without mass points. Note that the only way to select elements continuously is to select a dense subset of an interval.

Assume that seller $i$ sets a positive probability to a dense subset of $I=(p',p*)$, whereas seller $j$ does not select any prices in this interval. Let $p\in I$.

Note that $\beta_i$ is strictly decreasing in the neighborhood of $p$, and $\beta_j$ remains constant there. This is since $i$ select prices in the neighborhood of $p$ and $j$ not.

Note that from the definition it is known that $\alpha_i / \beta_j = \alpha_j / \beta_i$. Since $\beta_i$ is decreasing and $\beta_j$, and conclude that $\alpha_j$ is decreasing more rapidly then $\alpha_i$.

Therefore, for any price $p \in I$, the ratio of the parameters is $\alpha_i(p)>\alpha_j(p)$, except maybe the infimum of the interval.

Similarly, if both $i$ and $j$ do not select prices in an interval then $\alpha_i$ and $\alpha_j$ decrease in such interval at the same rate.

Let $\hat{p}$ be the infimum of an interval that is to the right of $I$, and is selected by $j$. For $\hat{p}$, the $\alpha$ parameters need to satisfy $\alpha_j(p) \geq \alpha_i(p)$. If both select this price - equality, if only $j$ does so - weak inequality.

Note that at $p$ the opposite inequality holds, and in all points between $p$ and $\hat{p}$, the parameter $\beta_j$ is decreasing less than $\beta_i$. Since all the $\alpha$'s and $\beta$'s change continuously everywhere except $P_M$, it is the case that $j$ cannot offer such prices.

Concluding, if a seller does not select an interval within $[P_L,P_M)$ she would not select any price above it, except possibly $P_M$, where the equation holds due to zero probability to sell to shoppers.
\qed

As shown in Lemma \ref{SearcherCondLemma}, for the reserve price to make sense, the following condition is sufficient:
The expected price of a seller is at least $P_L-c$.

Combining the lemmas the theorem \ref{symmod} is obtained.

\section{Asymmetric Model Proofs}

In this section I deal with the extended model and provide the results for the case when the smallest seller is not unique. Namely, this appendix shows the proof for theorem \ref{asymmod}-

\subsection{Similarity to old model}\label{Sec:symres}

The following lemmas carry over to this model with the same proof as before. The results in this subsection will be valid also for the single smallest firm case.

\begin{lemma}
No seller offers a price above $P_M$ in NE.
\end{lemma}

\begin{corollary}
Each searcher buys at the first store she visits.
\end{corollary}

\begin{lemma}
All sellers select $P_M$ as the supremum point of their strategy support.
\end{lemma}

\begin{lemma}
There are no mass points except possibly $P_M$
\end{lemma}

\begin{lemma}
Exists an interval $I$ such that the union of the seller strategies is contained in $I$ and dense in it.
\end{lemma}

\begin{corollary}
Exists an interval $I=[P_L,P_M]$, such that any NE strategy profile the sellers randomize continuously over $I$, and possibly some sellers set mass points at $P_M$.
\end{corollary}

\begin{lemma}
The previous lemma holds also for two sellers. Meaning - any interval has a point in distribution of at least two sellers.
\end{lemma}

\begin{corollary}
Any interval between $P_L$ and $P_M$ has points in the support of at least two sellers.
\end{corollary}

\subsection{Asymmetric Model Results}

However once one start dealing with the profits a difference exists:

Let the profit divided by the searchers fraction be denoted as the 'Profit per Branch' (PPB), and denote it as $\hat{\pi}$. The PPB measures the profit the seller gets divided by her store number.

\begin{lemma}
All sellers have the same PPB. That is, $\frac{\pi_i}{n_i}$ is equal to all sellers.
\end{lemma}

Similarly to the previous case, all sellers who have a mass point at $P_M$ have the probability of 0 to sell to shoppers at that price. This is since in any NE, due to undercutting, at least one seller will not have a mass point at that price and will select lower prices w.p. 1. Therefore PPB for sellers with a mass point at $P_M$ will be:
\begin{equation}
\hat{\pi}=\pi/n_i = P_M(1-\mu)/N
\end{equation}
Similarly, that would be the PPB if $P_M$ does not attract shoppers w.p.1.

Similarly to the symmetric model the only additional case that needed to be shown is: If $n-1$ sellers have the same PPB, the other seller cannot have a profit above them. Note that if only one seller has no mass point at $P_M$ then she must have a weakly higher PPB, as she can always deviate a strategy with a mass point at $P_M$. 

Suppose seller $i$ is the only seller who does not offer a mass point at $P_M$.
Let $p_i$ be the lowest (infimum if needed) price in the support of $i$. Additionally, seller $j$ with $n_j \leq n_i$. Since the lowest $n_i$ is not unique (from assumption) such $j$ always exists.

The profit of seller $i$ offering $p_i$ is (remember all searchers visit exactly one store):
\begin{equation}
\pi_i(p_i)=p_i((1-F_j(p)\prod_{k \neq i,j}(1-F_k(p_i))\mu+(1-\mu)n_i/N)
\end{equation}
The profit of seller $j$:
\begin{equation}
\pi_j(p_i)=p_i((1-F_i(p)\prod_{k \neq i,j}(1-F_k(p_i))\mu+(1-\mu)n_j/N)
\end{equation}

If one calculates the PPB for the two sellers, it will be equal to:
\begin{eqnarray}
\hat{\pi_i} = \frac{p_i((1-F_j(p)\prod_{k \neq i,j}(1-F_k(p_i))\mu}{n_i} + p_i \\
\hat{\pi_j} = \frac{p_i((1-F_i(p)\prod_{k \neq i,j}(1-F_k(p_i))\mu}{n_j} + p_i
\end{eqnarray}

And since $0=F_i(p_i) \leq F_j(p_i)$ one gets that the expression in the nominator of $\hat{\pi_j}$ is weakly higher than the of $\hat{\pi_i}$. 

Combined with the fact that $n_i \geq n_j$, it is clear that PPB of seller $j$ is weakly higher, which is a contradiction. \qed

Let $P_L$ be the lowest price seller $i$ can offer, while obtaining the PPB of $P_M$. $P_L$ needs to satisfy:

\begin{equation}
P_L (\mu / n_i +(1-\mu)/N) = P_M(1-\mu)/N)
\end{equation}

Note that $P_L$ is increasing in $n_i$. The point in this is only the sellers with the lowest store number will actually select a price at which they will sell to shoppers with certainty. Others will join in at a higher price, when the probability to attract shoppers is smaller than 1.

As before let $\alpha_i(p)$ be the probability of seller $i$ to sell to shoppers when offering price $p$.

Similarly to the identical sellers case, in order for seller $i$ to have the PPB of $P_M$ each price $p$ needs to satisfy:
\begin{equation}\label{PPB Equality}
p(\mu \alpha_i(p)/n_i + (1-\mu)/N) \leq P_M(1-\mu)/N
\end{equation}
As before, only seller with the highest $\alpha/n$ ($\alpha$ divided by the stores number) at a price will select it.

\begin{lemma}\label{None_Above}
Suppose seller $i$ has an interval $I$ in his support, and seller $j$ does not. This implies that in the support of seller $j$ there are no prices above $I$ except $P_M$.
\end{lemma}

This lemma is similar to the one in the symmetric case, however here it has stronger implications as the $P_L$ differs between sellers.

As before, the probability to attract shoppers is as follows:
\begin{equation}
\alpha_i(p) = (1-F_j(p))\prod_{k \neq i,j} (1-F_k(p))
\end{equation}

If one compares the PPB of seller $i$ and seller $j$:
\begin{eqnarray}
\frac{\alpha_i(p)}{n_i} = (1-F_j(p))\frac{\prod_{k \neq i,j} (1-F_k(p))}{n_i} \\
\frac{\alpha_j(p)}{n_j} = (1-F_i(p))\frac{\prod_{k \neq i,j} (1-F_k(p))}{n_j}
\end{eqnarray}
And it is known that the first is weakly larger than the second, as only the highest $\alpha/n$ can have the price in support. This implies that:

\begin{equation}
\frac{(1-F_j(p))}{n_i} \geq \frac{(1-F_i(p))}{n_j}
\end{equation}

It is shown that in $I$, seller $i$ has the highest $\alpha/n$. Assume that exists a price $p$ which is the lowest price above $i$ selected by seller $j$. This implies that at price $p'<P_M$ seller $j$ has the maximal $\alpha/n$. Note that since the prices between $I$ and $p$ were not selected, $F_j(p')=F_j(p)$ as $j$ does not select prices in between. $F_i$ however, had increased in $I$, it is known that seller $i$ has some mass over $I$. This implies that the inequality above holds also for $p'$, as an element on the right hand side was decreased, and at $p'$ it holds strictly:

\begin{equation}
\frac{(1-F_j(p'))}{n_i} > \frac{(1-F_i(p'))}{n_j}
\end{equation}

\begin{remark}
The reason that the result holds for prices below $P_M$ is that in the course of the proof a division by $\prod_{k \neq i,j} (1-F_k(p))$ was applied, and it needs to be positive. This happens at any price below $P_M$.
\end{remark}

The inequality above implies that seller $j$ cannot select the price $p'$, as seller $i$ has a higher $\alpha/n$, for $p<P_M$. Thus, if seller $j$ does not select some interval in the support union, she will not select any price above it except possibly $P_M$. \qed

\begin{corollary}
In any NE of the game, all sellers that do not have the lowest store number will select $P_M$ as a pure strategy. This is since such sellers cannot offer the price $P_L$ with sufficiently high PPB.
\end{corollary}

Combining the lemmas the proof of theorem \ref{asymmod} is obtained.

\section{Single Smallest Firm}

In this section I provide the proof for theorem \ref{singlefirm}

As before the smallest store number of a seller is denoted as $n_m$, and is the parameter of seller $m$. The next smallest size is denoted $n_j$ and is the parameter of seller $j$.

Here I deal with the extended model, and the case that a unique single firm exists. Thus, I provide here the required steps to the proof of theorem \ref{singlefirm}. The lemmas in section \ref{Sec:symres} hold with the same proof. The first difference occurs when dealing with PPB, and it is as follows:

\begin{lemma}
In any NE the PPB of all sellers except $m$ is equal to $P_M\frac{1-\mu}{N}$. The PPB of seller $m$ is strictly higher. Additionally, all sellers except $m$ have mass points at $P_M$.
\end{lemma}

The lowest price in the support union and denoted $P_L$.

PPB (profit per branch) of seller $i$ is the profit of seller $i$ divided by the $n_i$, which is a constant for a given seller.

Similarly to the previous case, some sellers mix, and at least two have at the support the price of $P_L$, where some seller $i \neq m$ is one of those. Comparing PPB of sellers $i$ and $m$ yield the following equations:

\begin{eqnarray}
PPB_i = P_L(\mu/n_i+(1-\mu)/N)\\
PPB_m = P_L(\mu/n_m+(1-\mu)/N)
\end{eqnarray}
Since $n_m< n_i$ from definition, it is clear that if $m$ offers this price she will have a higher PPB. This implies that if $m$ offers $P_L$ then she must have a higher PPB than all other sellers. 

If $m$ does not offer $P_L$, she can deviate to $P_L$ and increase her profit. Therefore, the PPB of seller $m$ is strictly larger than the PPB of some other sellers.

As the support supremum of all sellers is $P_M$, one of the two cases must hold (if all have mass points at $P_M$ undercutting is possible):

\begin{itemize}
\item At least two sellers do not have a mass point at $P_M$
\item A single seller does not have a mass point at $P_M$
\end{itemize}

In the first case all have equal PPB, as all offer the price $P_M$ and by this price attract no shoppers, having the PPB of $P_M(1-\mu)/N$. In the second case, all but one seller have the mentioned PPB, when the last seller have a weakly higher one. From the previous steps it is known that the second case holds, and the seller without a mass point at $P_M$ is seller $m$, and she has the highest PPB. \qed.

From lemma \ref{None_Above}, only sellers with $n_j$ can have prices below $P_M$ in support. As at least two sellers need to mix over the entire interval, and sellers with $n_j$, and all sellers except $m$ have the same PPB, all the sellers with parameter $n_j$ can mix over the entire interval, have a cutoff price or select the reserve price purely. Additionally, when mixing the same distr. must be used. Note that at least one of the sellers, denoted $j$, must mix over the entire interval, as the support must be covered by at least two sellers.

The last point needed to be show is that $m$ has higher probability to offer discount. Namely, $F_j < F_m$ in $(P_L,P_M)$. Remember that:
\begin{equation}
\alpha_i(p)= \prod_{k \neq i}(1-F_k(p)=\frac{\prod\left(1-F_k(p)\right)}{1-F_i(p)}
\end{equation}
Therefore, the larger $\alpha$ will have the larger $F$. If one writes the profit expressions for the sellers $j$ and $m$:
\begin{eqnarray}
\pi_m(p)=p(\alpha_m(p)+\frac{n_m(1-\mu)}{N}) = P_L\left(\mu+\frac{n_m(1-\mu)}{N}\right)\\
\pi_m(p)=p(\alpha_j(p)+\frac{n_j(1-\mu)}{N}) = P_M\left(\frac{n_j(1-\mu)}{N}\right)
\end{eqnarray}

For simplicity, I denote $Src_i=(1-\mu)n_i/N$ 

Extracting the expressions for $\alpha_m$ and $\alpha_j$ the following equation is obtained:
\begin{eqnarray}
\alpha_m = \frac{(P_L-p)Src_m+P_L\mu}{p \mu}\\
\alpha_j = \frac{P_M-p}{p \mu}Src_j
\end{eqnarray}

Obtaining the derivatives and comparing, one gets that $\alpha_m > \alpha_j$ for prices in $(P_L,P_M)$, as required.

\end{appendix}

\end{document}